\documentclass[12pt]{article}

\usepackage{mystyle}
\usepackage{hyperref}
\hypersetup{
    colorlinks=true, 
    linktoc=all,     
    linkcolor=blue,  
    citecolor=blue,
    linktocpage,
}
\usepackage[margin=3cm]{geometry}

\title{On the Spectral Properties of Symmetric Functions}

\author{
Anil Ada\footnote{Department of Computer Science, Carnegie Mellon University. Email: {\small {\tt aada@cs.cmu.edu}.} }   \and  
Omar Fawzi\footnote{LIP, \'Ecole Normale Sup\'erieure de Lyon. Email: {\small {\tt omar.fawzi@ens-lyon.fr}.} } \and
Raghav Kulkarni\footnote{Chennai Mathematical Institute. Email: {\small {\tt kulraghav@gmail.com}.} } 
} 

\begin{document}

\maketitle

\begin{abstract}
We characterize the approximate monomial complexity, sign monomial complexity, and the approximate $L_1$ norm of symmetric functions in terms of simple combinatorial measures of the functions. Our characterization of the approximate $L_1$ norm solves the main conjecture in \cite{AFH12}. As an application of the characterization of the sign monomial complexity, we prove a conjecture in \cite{SZ09} and provide a characterization for the unbounded-error communication complexity of symmetric-xor functions.
\end{abstract}

\section{Introduction}

Understanding the structure and complexity of Boolean functions $f:\{0,1\}^n \to \{0,1\}$ is a main goal in computational complexity theory. Fourier analysis of Boolean functions provide many useful tools in this study. Natural Fourier analytic properties of a Boolean function can be linked to the computational complexity of the function in various settings like circuit complexity, communication complexity, decision tree complexity, learning theory, etc.

In this paper, our focus is on trying to understand the Fourier analytic (i.e. spectral) properties of \emph{symmetric} functions, which are Boolean functions such that permuting the input bits does not change the output. Many basic and fundamental functions like $\f{and}, \f{or}, \f{majority}, \f{mod}_m$ are symmetric, and having a full understanding of the spectral properties of symmetric functions is a natural goal.

Some of the important spectral properties of Boolean functions are the degree (the largest degree of a monomial with non-zero Fourier coefficient), the monomial complexity (the number of non-zero Fourier coefficients), and the Fourier $L_p$ norms. Often, the degree or the monomial complexity of a Boolean function does not give us useful information, so we study approximate versions like $\eps$-approximate degree (the minimum degree of a polynomial that point-wise approximates the function) and sign degree (the minimum degree of a polynomial that sign represents the function). These measures have found numerous applications in computational complexity theory. 

Some earlier results on the spectral properties of symmetric functions include the characterization of sign degree \cite{ABFR94}, approximate degree \cite{Pat92}, and Fourier $L_1$ norm \cite{AFH12}. 

Our main results are as follows.
\begin{itemize}
    \item Theorem~\ref{thm:main1}: characterization of approximate monomial complexity of symmetric functions.
    \item Theorem~\ref{thm:main2}: characterization of sign monomial complexity of symmetric functions.
    \item Corollary~\ref{cor:infinity-norm}: a lower bound on the $L_\infty$ norm of symmetric functions.
    \item Theorem~\ref{thm:main4}: characterization of approximate $L_1$ norm of symmetric functions. This solves the main conjecture of \cite{AFH12}.
\end{itemize}

Our results have the following applications in communication complexity. 
\begin{itemize}
    \item Theorem~\ref{thm:main3}: characterization of the unbounded-error communication complexity of symmetric-xor functions. This resolves a conjecture of \cite{SZ09}. This result was obtained independently by Hatami and Qiang \cite{HQ17}.
    \item Theorem~\ref{thm:main5}: verifying the log approximation rank conjecture for symmetric-xor functions.
\end{itemize}

To prove these results, we make use of (i) the close connections between Boolean functions and their corresponding two-party xor functions (Proposition~\ref{prop:ftoF}), and (ii) the known bounds on the approximate rank and the sign rank of two-party symmetric-and functions (Theorem~\ref{thm:raz03} and Theorem~\ref{thm:she12}). We transform these results on two-party symmetric-and functions to the setting of symmetric xor-functions via reductions. 

\section{Preliminaries}

\subsection*{General notation}

We use $[n]$ to denote the set $\{1,2,\ldots, n\}$. All the logarithms are base 2. For $x \in \{0,1\}^n$, $|x|$ denotes the Hamming weight of $x$, i.e., $|x| = \sum_i x_i$. For $b \in \{0,1\}$, $\neg b$ denotes negation of $b$. Given $x$ and $y$ in $\{0,1\}^n$, $x \wedge y$ denotes the $n$-bit string obtained by taking the coordinate-wise \emph{and} of $x$ and $y$. Similarly, $x \xor y$ denotes the $n$-bit string obtained by taking the coordinate-wise \emph{xor} of $x$ and $y$.

A Boolean function $f:\{0,1\}^n \to \{0,1\}$ is called \emph{symmetric} if the function's output does not change when we permute the input variables. When $f$ is symmetric, we'll use $f$ to also denote the corresponding function $f:\{0,1,\ldots,n\} \to \{0,1\}$ with the understanding that $f(|x|) = f(x)$. We define $r_0(f)$ and $r_1(f)$ 
\begin{align*}
r_0(f) &\defeq \min \{ r \leq \ceil{n/2} : f(i) = f(i+2) \text{ for all } i \in [r, \ceil{n/2}-1] \} \\
r_1(f) &\defeq \min \{ r \leq \floor{n/2} - 1 : f(i) = f(i+2) \text{ for all } i \in [\ceil{n/2}, n - r - 2] \}
\end{align*}
Note that we have $f(i) = f(i+2)$ for all $i \in [r_0(f), n-r_1(f)-2]$. Then $r(f) = \max\{r_0(f),r_1(f)\}$. 
Also, we let 
\[
    \lambda(f) \defeq |\{i : f(i) \neq f(i+1)\}|,
\] 
and 
\[
    \rho(f) \defeq |\{i : f(i) \neq f(i+2)\}|.
\]
When the function is clear from the context, we may drop the $f$ from this notation.

\subsection*{Fourier analysis}

Let $f:\{0,1\}^n \to \{0,1\}$ be a Boolean function. We view $f$ as residing in the $2^n$-dimensional vector space of real-valued functions $\{\phi:\{0,1\}^n \to \R\}$. We equip this vector space with the inner product $\ip{\phi, \phi'} \defeq \exd{}{\phi(x)\phi'(x)}$, where $x$ is uniformly distributed over $\{0,1\}^n$. For each $S \subseteq [n]$, define the function
\[
    \chi_S(x) \defeq (-1)^{\sum_{i \in S} x_i}.
\]
We refer to these functions as \emph{characters} or \emph{monomials}. It is easy to check that the set $\{\chi_S : S \subseteq [n]\}$ forms an orthonormal basis. Therefore every function $\phi$ (including every Boolean function) can be written as $\sum_{S \subseteq [n]} \fc{\phi}(S) \chi_S$, where $\fc{\phi}(S) = \ip{\phi, \chi_S}$ are the real-valued coefficients, called the \emph{Fourier coefficients}. This way of expanding $\phi$ is called the \emph{Fourier expansion} of $\phi$. 

The \emph{degree} of a function $\phi$ is defined as $\deg(\phi) \defeq \max\{|S|: \fc{\phi}(S) \neq 0\}$ and the \emph{monomial complexity} is defined as $\mon(\phi) \defeq |\{S: \fc{\phi}(S) \neq 0\}|$. We also define the Fourier $p$-norm: 
\[
    \norm{\fc{\phi}}_p \defeq \left( \sum_S |\fc{\phi}(S)|^p \right)^{1/p}.
\]
The Fourier infinity norm is defined to be $\norm{\fc{\phi}}_\infty = \max_S |\fc{\phi}(S)|$.

For symmetric functions, \cite{AFH12} characterized the Fourier $1$-norm in terms of $r(f)$.
\begin{theorem}[\cite{AFH12}] \label{thm:afh12}
Let $f:\{0,1\}^n \to \{0,1\}$ be a symmetric function. When $r(f) > 1$, we have
\[
    \log \norm{\fc{f}}_1 = \Theta\left(r(f) \log\left(\frac{n}{r(f)}\right)\right) \ .
\]
\end{theorem}

\subsection*{Matrix analysis}

Let $M \in \R^{k \times k}$ be a real-valued matrix, with singular values $\sigma_1, \sigma_2, \ldots, \sigma_k \geq 0$. The rank of $M$, denoted $\rank(M)$, is the number of non-zero singular values. The Schatten $p$-norm is defined as follows:
\begin{align*}
    \norm{M}_p  & \defeq \left( \sum_{i = 1}^k \sigma_i^p \right)^{1/p}, \\
    \norm{M}_\infty & \defeq \sigma_1.
\end{align*}
We then define 
\begin{alignat*}{3}
    \text{trace norm:} & \quad \norm{M}_{\tr} &\; \defeq \; & \norm{M}_1\\ 
    \text{Frobenius norm:} & \quad \norm{M}_{\Fr} & \; \defeq \; &  \norm{M}_2\\ 
    \text{spectral norm:} & \quad \norm{M} & \; \defeq \; & \norm{M}_\infty 
\end{alignat*}

Given two matrices $M$ and $M'$, we write $M = M'$ if one can be obtained from the other after reordering the rows and/or the columns.

\subsection*{Approximation theory}

Throughout the paper, $\eps$ denotes any constant in $[0, 1/2)$. Given $f:\{0,1\}^n \to \{0,1\}$, we say that $\phi:\{0,1\}^n \to \mathbb{R}$ $\eps$-approximates $f$ if for all $x \in \{0,1\}^n$, $|\phi(x) - f(x)| \leq \eps$. Then the $\eps$-approximate monomial complexity of $f$, denoted by $\mon_\eps(f)$, is defined as the minimum monomial complexity of a function that $\eps$-approximates $f$. Similarly we define $\norm{\fc{f}}_{1,\eps}$. For a matrix $M$, $\rank_\eps(M)$ and $\norm{M}_{\tr,\eps}$ are defined as the minimum rank and the minimum trace norm respectively, of a matrix that $\eps$-approximates $M$ entry-wise.

Given $f:\{0,1\}^n \to \{0,1\}$, we say that $\phi:\{0,1\}^n \to \mathbb{R}$ sign-represents $f$ if for all $x$ such that $f(x) = 1$, $\phi(x) > 0$, and for all $x$ such that $f(x) = 0$, $\phi(x) < 0$. The sign monomial complexity of $f$, denoted $\signmon(f)$, is defined to be the minimum monomial complexity of a function $\phi$ that sign represents $f$. For a matrix $M$ with entries in $\{0,1\}$, we similarly define $\signrank(M)$.

The following proposition provides a relationship between the approximate trace norm and the approximate rank:
\begin{proposition}[Folklore] \label{prop:tracetorank}
Let $M \in \{0,1\}^{k \times k}$. Then,
\[
    \rank_\eps(M) \geq  \left( \frac{\norm{M}_{\tr,\eps}}{k (1+\eps)}\right)^2.
\]
\end{proposition}
\begin{proof}
Let $M'$ be a matrix that entry-wise $\eps$-approximates $M$, and $\rank(M') = \rank_\eps(M)$. Then
\[
    \norm{M}_{\tr, \eps} \leq \norm{M'}_{\tr} \overset{(*)}{\leq} \norm{M'}_{\Fr} \sqrt{\rank(M')} \leq k(1+\eps) \sqrt{\rank(M')} = k(1+\eps) \sqrt{\rank_\eps(M)},
\]
where we used the Cauchy-Schwarz inequality for $(*)$. 
\end{proof}

Bruck and Smolensky \cite{BS92} provided an upper bound on the sign monomial complexity of a Boolean function in terms of its Fourier $1$-norm. In fact, their proof gives an upper bound on the approximate monomial complexity too. 
\begin{theorem}[\cite{BS92}] \label{thm:bs92}
For any $f:\{0,1\}^n \to \{0,1\}$,
\[
    \mon_\eps (f) \leq \frac{4n}{\eps^2} \norm{\fc{f}}_1^2.
\]
\end{theorem}

Bruck \cite{Bru90} gave a lower bound on the sign monomial complexity of a Boolean function in terms of the Fourier infinity norm of $f$.
\begin{theorem}[\cite{Bru90}] \label{thm:bru90}
Let $f:\{0,1\}^n \to \{0,1\}$ be a Boolean function. Then
\[
    \signmon(f) \geq \frac{1}{\norm{\fc{f}}}_\infty.
\]
\end{theorem}

\subsection*{Two-party functions}

A capital function name will refer to a function with two inputs, e.g., $F : \calX \times \calY \to \{0,1\}$ where $\calX$ and $\calY$ are some finite sets. We'll abuse notation and also use $F$ to denote the $|\calX|$ by $|\calY|$ matrix corresponding to $F$, i.e., the $(x,y)$'th entry of the matrix contains the value $F(x,y)$. It will always be clear from the context whether $F$ refers to a function or a matrix. 

Given $f: \{0,1,\ldots,n\} \to \{0,1\}$, we'll define $f_i : \{0,1,\ldots,i\} \to \{0,1\}$ by $f_i(j) = f(j)$. We denote by $F^{\wedge}_{n,f}: \{0,1\}^n \times \{0,1\}^n \to \{0,1\}$ the communication function such that $F^{\wedge}_{n,f}(x,y) = f(|x\wedge y|)$. We use the notation $F^{\wedge}_{n,k,f_k}$ when the inputs $x$ and $y$ are promised to satisfy $|x| = |y| = k$. Similarly, we define $F^{\xor}_{n,f}$ and $F^{\xor}_{n,k,f_{2k}}$, for $k \leq n/2$.

In an important paper, Razborov \cite{Raz03} gave close to tight lower bounds on the randomized communication complexity of $F^{\wedge}_{n,f}$ where $f$ is a symmetric function. His main result can be stated as a lower bound on the approximate trace norm of a certain submatrix of $F^{\wedge}_{n,f}$:
\begin{theorem}[\cite{Raz03}] \label{thm:raz03}
For $k \leq n/4$, let $f: \{0,1,\ldots, k\} \to \{0,1\}$. If for some $\ell \leq k/4$ we have $f(\ell - 1) \neq f(\ell)$, then 
\[
    \norm{F^{\wedge}_{n,k,f}}_{\tr,1/4} \geq {n \choose k} e^{\Omega(\sqrt{k\ell})}.
\]
\end{theorem}

We'll also need a result from Sherstov \cite{She12} that gives essentially tight lower bounds on the sign-rank of all \emph{symmetric-and} functions (see Section~\ref{sec:communication} for this result's relation to communication complexity). 

\begin{theorem}[\cite{She12}] \label{thm:she12}
Let $f : \{0,1,\ldots, n\} \to \{0,1\}$. Then
\[
    \signrank(F^{\wedge}_{n,f}) \geq 2^{\Omega(\lambda(f) / \log^5 n)}.
\]
\end{theorem}

Our main interest in 2-party functions is due to the tight links between the Fourier analytic properties of a Boolean function $f:\{0,1\}^n \to \{0,1\}$ and the matrix analytic properties of $F = F^\xor_{n,f}$.

\begin{proposition}[Folklore] \label{prop:ftoF}
Let $f: \{0,1\}^n \to \{0,1\}$ be any function and let $F = F^\xor_{n,f}$. Then
\begin{enumerate}[(a)]
    \item $\mon(f) = \rank(F)$,
    \item $\mon_\eps(f) \geq \rank_\eps(F)$
    \item $\signmon(f) \geq \signrank(F)$,
    \item $2^n \norm{\fc{f}}_\infty = \norm{F}$,
    \item $2^n \norm{\fc{f}}_1 = \norm{F}_{\tr}$,
    \item $2^n \norm{\fc{f}}_{1,\eps} = \norm{F}_{\tr, \eps}$.
\end{enumerate}
\end{proposition}
%

\section{Main Results}

\begin{theorem} \label{thm:main1}
Let $f:\{0,1\}^n \to \{0,1\}$ be a symmetric function. Then,
\[
    \Omega(r(f)) \leq \log \mon_{1/4}(f) \leq O\left(r(f) \log\left(\frac{n}{r(f)}\right)\right).
\]
\end{theorem}
\begin{proof}\ \\
\noindent
\emph{Lower bound:}

We first note that we may assume that $r(f) = r_0(f)$. In fact, if $r(f) = r_1(f)$, then we can consider the function $f'$ defined as $f'(i) = f(n-i)$. We note that $\mon_{\eps}(f) = \mon_{\eps}(f')$. To see this, given a function $g$ approximating $f$ with $\| f - g \|_{\infty} \leq \eps$, the function $g'$ defined by $g'(x_1, \dots, x_n) = g(1-x_1, \dots, 1-x_n)$ satisfies $\| f' - g' \|_{\infty} \leq \eps$ and $|\fc{g'}(S)| = 2^{-n} |\sum_{x} g(x) \chi_{S}(1-x_1, \dots, 1-x_n)| = |\fc{g}(S)|$. This shows that $\mon_{1/4}(f) = \mon_{1/4}(f') \geq r_0(f')$. But we have $f(n-r_1(f)-1) \neq f(n-r_1(f)+1)$, i.e., $f'(r_1(f)-1) \neq f'(r_1(f)+1)$ (except if $r_1(f) = 0$, but this case is simple). This implies that $r_0(f') \geq r_1(f)$. 

For the remainder of the proof, we assume there is an $s \in \{1, \dots, \ceil{n/2}\}$ such that $f(s-1) \neq f(s+1)$ and $s \geq r(f)$.

In light of Proposition~\ref{prop:ftoF}, part (b), our goal will be to show a lower bound on $\rank_\eps(F^\xor_{n,f})$. For any $k,t$ such that $2k+t \leq n$, we define the submatrix $F^\xor_{n-t,k,f_{2k}}$ of $F^\xor_{n,f}$ of size $\binom{n-t}{k} \times \binom{n-t}{k}$ by $F^\xor_{n-t,k,f_{2k}}(x,y) = f(|x \xor y| + t)$ for all $x,y \in \{0,1\}^{n-t}$. Note that this is for example the submatrix obtained by considering all the bitstrings $x',y' \in \{0,1\}^n$ for which the first $t$ bits are set to one and among the remaining $n-t$ bits, exactly $k$ are set to $1$.

Observe that $|x \xor y| = |x| + |y| - 2|x \wedge y|$. In particular, when $|x| = |y| = k$, $|x \xor y| = 2k - 2|x \wedge y|$. This means that 
\[
F^\xor_{n-t,k,f_{2k}} = F^{\wedge}_{n-t,k,f'_{k}},
\]
where 
\[
f'_{k}(i) = f_{2k}(2k - 2i + t) \text{ for } i \in \{0,1,\ldots , k\}.
\]
Thus, we'll show a lower bound on the approximate-rank of $F^{\wedge}_{n-t,k,f'_{k}}$. To do this, first we'll use Proposition~\ref{prop:tracetorank}, and show a lower bound on the approximate-trace norm. To show a lower bound on the approximate-trace norm, we'll use Theorem~\ref{thm:raz03} and the fact that
\[
f_{2k}(s - 1) \neq f_{2k}(s+1) \implies f'_k\left(k + \frac{t-(s-1)}{2}\right) \neq f'_k\left(k + \frac{t-(s+1)}{2}\right).
\]
In other words, our choice for $\ell$ in Theorem~\ref{thm:raz03} will be $k + \frac{t-s-1}{2}$. Let's now specify $k$ and $t$. Note that we should make sure that $t - s - 1$ is even. We distinguish two cases depending on whether $s \leq 3(n-1)/8$ or not.

If $s \leq 3(n-1)/8$, then we simply set $t = 0$ if $s$ is odd and $t = 1$ if $s$ is even. Then we let $k = \lfloor \frac{2s}{3} \rfloor$. Since $s \leq 3(n-1)/8$, it is easy to check that $k \leq (n-t)/4$ and $\ell \leq k/4$ as required by Theorem~\ref{thm:raz03}. So we have
\[
\norm{F^\xor_{n-t,k,f_{2k}}}_{\tr, 1/4} = \norm{F^\wedge_{n-t,k,f'_k}}_{\tr, 1/4} \geq {n-t \choose k} e^{\Omega(\sqrt{k\ell})},
\]
which, by Proposition~\ref{prop:tracetorank} and our choices for $k$ and $\ell$, implies
\[
\rank_{1/4}(F^\xor_{n-t,k,f_{2k}}) \geq e^{\Omega(\sqrt{k \ell})} = e^{\Omega(s)}.
\]

In the case $s > 3(n-1)/8$, we set $t = \floor{n/4}$ or $t = \floor{n/4} - 1$ depending on the parity of $s$, and $k = \floor{2(s-1-n/4)/3}$. We then have $k \leq n/6$ using the fact that $s \leq n/2+1$. As $t \leq n/4$, this implies that $k \leq \frac{n-t}{4}$. On the other hand, we have $k = \Omega(n)$. Now recall that $\ell = k + \frac{t - s - 1}{2}$. But $\frac{s+1-t}{2} \geq 3k/4$ which implies that $\ell \leq k/4$. In addition, as $s > 3(n-1)/8$, we also have $\ell = \Omega(n)$. As a result, we can apply Theorem~\ref{thm:raz03} and obtain
\[
\rank_{1/4}(F^\xor_{n-t,k,f_{2k}}) \geq e^{\Omega(\sqrt{k \ell})} = e^{\Omega(s)}.
\]
Using Proposition~\ref{prop:ftoF} part (b), we obtain the desired result.

\noindent
{\em Upper bound:} 

Using Theorem~\ref{thm:bs92}, we have $\mon_{1/4}(f) \leq 64n \| \fc{f} \|_1^2$. Taking the logarithm and using Theorem~\ref{thm:afh12} we get the desired result.
\end{proof}

\begin{theorem} \label{thm:main2}
Let $f:\{0,1\}^n \to \{0,1\}$ be a symmetric function. Then,
\[
    \Omega(\rho(f)/\log^5 n) \leq \log \mon_\pm(f) \leq O(1+\rho(f)\log n).
\]
\end{theorem}
\begin{proof}\ \\
\noindent
{\em Lower bound:} 

First, we'll assume $|\{i \in [2,2n/3] : f(i) \neq f(i-2) \text{ and $i$ is even}\}|$ is a constant fraction of $\rho(f)$. At the end of the proof, we give an argument for when this is not true. 

In light of Proposition~\ref{prop:ftoF}, part (c), our goal is to show that
\begin{equation} \label{eq:main2eq1}
    \log \signrank(F^\xor_{n,f}) = \Omega(\rho(f)/\log^5 n).
\end{equation}
Since $F^\xor_{n,n/3,f_{2n/3}}$ is a submatrix of $F^\xor_{n,f}$, it suffices to show a lower bound on the sign-rank of $F^\xor_{n,n/3,f_{2n/3}}$. As in the proof of Theorem~\ref{thm:main1}, 
\[
    F^\xor_{n,n/3,f_{2n/3}} = F^\wedge_{n,n/3,f'_{n/3}}, 
\]
where 
\[
    f'_{n/3}(i) = f_{2n/3}(2n/3 - 2i) \quad \text{ for } i \in \{0,1,\ldots , n/3\}.
\] 
From the assumption we made at the beginning of the proof, we know that $\lambda(f'_{n/3}) = \Omega(\rho(f))$. By Theorem~\ref{thm:she12}, we know that 
\[
    \log \signrank(F^\wedge_{n/3,f'_{n/3}}) = \Omega(\lambda(f'_{n/3}) / \log^5 (n/3)).
\]
We show that the above implies
\begin{equation} \label{eq:main2eq2}
    \log \signrank(F^\wedge_{n,n/3,f'_{n/3}}) = \Omega(\lambda(f'_{n/3}) / \log^5 (n/3)),
\end{equation}
by showing that $F^\wedge_{n/3,f'_{n/3}}$ is a submatrix of $F^\wedge_{n,n/3,f'_{n/3}}$, as follows. Given $x, y \in \{0,1\}^{n/3}$, construct (by padding $x$ and $y$ appropriately with $2n/3$ bits each) $x', y' \in \{0,1\}^n$ with the property that the Hamming weights $|x'| = |y'| = n/3$ and the strings don't intersect at indices $i \in [n/3+1,n]$. Clearly the mappings $x \mapsto x'$ and $y \mapsto y'$ are injective, and $|x \wedge y| = |x' \wedge y'|$. So $F^\wedge_{n/3,f'_{n/3}}$ is a submatrix of $F^\wedge_{n,n/3,f'_{n/3}}$. This establishes \eqref{eq:main2eq2}, and therefore \eqref{eq:main2eq1}. This completes the proof for the case when $|\{i \in [2,2n/3] : f(i) \neq f(i-2) \text{ and $i$ is even}\}|$ is a constant fraction of $\rho(f)$.


If the changes $f(i) \neq f(i-2)$ are happening mostly at odd indices $i$, then consider the restriction of $f$ in which one input variable is set to 1. If $f'$ is this restriction, then $F^\xor_{n-1,f'}$ is a submatrix of $F^\xor_{n,f}$ and therefore it suffices to show a lower bound on $\signrank(F^\xor_{n-1,f'})$.

If $|\{i \in [2,2n/3] : f(i) \neq f(i-2)\}|$ is not a constant fraction of $\rho(f)$, then consider the function $f'$ defined as $f'(i) = f(n-i)$. This $f'$ is such that $|\{i \in [2,2n/3] : f'(i) \neq f'(i-2)\}|$ a constant fraction of $\rho(f) = \rho(f')$. Furthermore, note that $F^\xor_{n,f} = F^\xor_{n,f'}$ as one is obtained from the other by rearranging the columns. Therefore it suffices to show a lower bound on $\signrank(F^\xor_{n,f'})$.

\noindent
{\em Upper bound:} 

We'll prove by induction on $\rho(f)$ that $\signmon(f) \leq (n+2)^{\rho(f)}$. If $\rho(f)=0$ then $f$ is either a constant function or a parity function (parity or its negation), and so can be represented exactly using at most two non-zero Fourier coefficients. We also have to explicitly prove the $\rho(f) = 1$ case. Let's consider the function $f$ with $f(i) = \f{parity}(i)$ for $i < t$ and $f(i) = 0$ for $i \geq t$, for some $t$. Observe that the following polynomial sign represents $f$:
\[
    (2t-0.1) (-1)^{x_1+x_2+\ldots +x_n} + ((-1)^{x_1} + (-1)^{x_2} + \cdots + (-1)^{x_n} - n).
\]
So $\signmon(f) \leq n+2$. By slightly modifying the above polynomial, we can sign represent any function that behaves like a parity function (parity or its negation) for $i \leq t$ and behaves like a constant function for $i > t$. We can also sign represent any function that behaves like a constant function for $i \leq t$ and behaves like a parity function for $i > t$. These are the only cases to consider for $\rho(f) = 1$.

Now suppose $\rho(f) > 1$. Let $j$ be the largest index such that $f(j) \neq f(j-2)$. Let $f'$ be the function obtained from $f$ as follows: $f'(i) = f(i)$ for $i < j$ and $f'(i) = f'(i-2)$ for $i \geq j$. Observe that $\rho(f') = \rho(f) - 1$. Let $p'$ be a sign representing polynomial for $f'$ with at most $(n+2)^{\rho(f')} = (n+2)^{\rho(f)-1}$ monomials. Let $f''$ be the function obtained from $f$ as follows: $f''(i) = 1$ for $i < j$ and either $f''(i) = \f{parity}(i)$ for $i \geq j$ or $f''(i) = \neg\f{parity}(i)$ for $i \geq j$. Observe that $\rho(f'') = 1$, and so it has a sign representing polynomial $p''$ with at most $n+2$ monomials. The functions $f'$ and $f''$ are constructed in a way so that the product $p'\cdot p''$ sign represents $f$ (in particular, the choice of $f''(i) = \f{parity}(i)$ or $f''(i) = \neg\f{parity}(i)$ for $i \geq j$ is made accordingly). Therefore $\signmon(f) \leq (n+2)^{\rho(f)}$.

\end{proof}

As a corollary to the upper bound above, we can give a lower bound on the Fourier infinity norm of a symmetric function.

\begin{corollary} \label{cor:infinity-norm}
Let $f:\{0,1\}^n \to \{0,1\}$ be a symmetric function. Then
\[
    \norm{\fc{f}}_\infty \geq \frac{1}{(n+2)^{\rho(f)}}.
\]
\end{corollary}
\begin{proof}
From the proof of Theorem~\ref{thm:main2}, we have $\signmon(f) \leq (n+2)^{\rho(f)}$. Combining this with Theorem~\ref{thm:bru90} gives the desired bound.
\end{proof}

We now prove the main conjecture of \cite{AFH12}.

\begin{theorem} \label{thm:main4}
Let $f:\{0,1\}^n \to \{0,1\}$ be a symmetric function. Then,
\[
    \Omega(r(f)) - \frac{1}{2} \log n \leq \log \norm{\fc{f}}_{1,1/5} \leq \log \norm{\fc{f}}_1 \leq O\left(r(f) \log\left(\frac{n}{r(f)}\right)\right).
\]
\end{theorem}
\begin{proof}
The upper bound is in Theorem~\ref{thm:afh12}. For the lower bound, let $g$ be such that $\norm{\fc{f}}_{1,1/5} = \norm{\fc{g}}_1$. 
Applying Theorem~\ref{thm:bs92} with $\eps = 1/20$, we get $\norm{\fc{g}}_1 \geq \frac{\eps}{2\sqrt{n}} \sqrt{\mon_{1/20}(g)}$. By the triangle inequality, we have $\mon_{1/20}(g) \geq \mon_{1/4}(f)$. Thus,
\begin{align*}
\log \norm{\fc{f}}_{1,1/5} \geq \frac{1}{2} \log \mon_{1/4}(f) - \frac{1}{2} \log n - \log 40 \ .
\end{align*}
To conclude, it suffices to use Theorem~\ref{thm:main1} 
\end{proof}

\section{Applications to Communication Complexity}

\label{sec:communication}

We denote by $\bfR_\eps(F)$ the $\eps$-error randomized communication complexity of $F$. In this model, the players are allowed to share randomness and for all inputs, they are required to output the correct answer with probability at least $1-\eps$. We'll think of $\eps$ as some constant less than $1/2$.

Here we'll also be interested in the \emph{unbounded-error randomized communication complexity} of a function $F: \calX \times \calY \to \{0,1\}$, denoted $\mathbf{U}(F)$. In this model, the players have private randomness and the only requirement from the protocol is that for all inputs, it gives the correct answer with probability greater than 1/2. Notice that achieving error probability 1/2 is trivial: just output a random bit. Also, note that there is no requirement that the success probability be bounded away from 1/2, e.g., the success probability could be $1/2 + 1/2^n$. This makes the model quite powerful and proving lower bounds much harder. It was shown in \cite{PS86} that
\[
\bfU(F) = \log_2 \signrank(F) \pm O(1).
\]

In a remarkable paper \cite{For02}, Forster was able to prove a lower bound on the unbounded error communication complexity of a function using the function's spectral norm. In particular he was able to show a linear lower bound for the \emph{inner-product} function. 
Building on Forster's work, Sherstov \cite{She12} gave essentially tight lower bounds on the unbounded error communication complexity of all \emph{symmetric-and} functions $F^{\wedge}_{n,f}$ (see Theorem~\ref{thm:she12}).

In \cite{SZ09}, Shi and Zhang conjecture that the unbounded error communication complexity of a \emph{symmetric-xor function} $F^{\xor}_{n,f}$ is characterized by $\rho(f)$. We prove this conjecture. First, the proof Theorem~\ref{thm:main2} allows us to bound the sign-rank of \emph{symmetric-xor} functions.

\begin{theorem} \label{thm:main3}
Let $f:\{0,1\}^n \to \{0,1\}$ be a symmetric function. Then,
\[
    \Omega(\rho(f)/\log^5 n) \leq \log \signrank(F^{\xor}_{n,f}) \leq O(1+\rho(f)\log n).
\]
\end{theorem}
\noindent
This immediately implies:

\begin{corollary} 
Let $f:\{0,1\}^n \to \{0,1\}$ be a symmetric function. Then,
\[
    \Omega(\rho(f)/\log^5 n) \leq \log \bfU(F^{\xor}_{n,f}) \leq O(1+\rho(f)\log n).
\]
\end{corollary}

The second application is related to the \emph{Log Approximation Rank Conjecture}, which is the randomized communication complexity analog of the famous Log Rank Conjecture. The Log Approximation Rank Conjecture states that there is a constant $c$ such that for any 2-party function $F$,
\[
    \log \rank_{\eps'}(F) \leq \bfR_\eps(F) \leq \log^c \rank_{\eps'}(F).
\]
Here, the lower bound is well-known to be true for all functions, so the conjecture is about establishing the upper bound. This has been done by Razborov \cite{Raz03} for symmetric-and functions $F^{\wedge}_{n,f}$. We show that the conjecture holds also for symmetric-xor functions $F^{\xor}_{n,f}$.

\begin{theorem} \label{thm:main5}
There are constants $\eps, \eps',c > 0$ such that for any two-party function $F^{\xor}_{n,f}$, where $f$ is symmetric, we have
\[
    \bfR_\eps(F^{\xor}_{n,f}) \leq \log^c \rank_{\eps'}(F).
\]
\end{theorem}
\begin{proof}
By Proposition 3.4 of \cite{SZ09}, we know that 
\[
    \bfR_\eps(F^{\xor}_{n,f}) \leq O(r(f) \log^2 r(f) \log \log r(f)).
\]
The proof of Theorem~\ref{thm:main1} allows us to conclude that 
\[
    \log \rank_{1/4}(F^{\xor}_{n,f}) \geq \Omega(r(f)).
\]
Combining the two bounds proves the result.
\end{proof}

\bibliographystyle{alpha}
\bibliography{thesis}

\begin{thebibliography}{ABFR94}

\bibitem[ABFR94]{ABFR94}
James Aspnes, Richard. Beigel, Merrick Furst, and Steven Rudich.
\newblock The expressive power of voting polynomials.
\newblock {\em Combinatorica}, 14(2):135--148, 1994.

\bibitem[AFH12]{AFH12}
Anil Ada, Omar Fawzi, and Hamed Hatami.
\newblock Spectral norm of symmetric functions.
\newblock In {\em APPROX-RANDOM}, pages 338--349, 2012.

\bibitem[Bru90]{Bru90}
Jehoshua Bruck.
\newblock Harmonic analysis of polynomial threshold functions.
\newblock {\em SIAM Journal of Discrete Mathematics}, 3:168--177, 1990.

\bibitem[BS92]{BS92}
Jehoshua Bruck and Roman Smolensky.
\newblock Polynomial threshold functions, ac0 functions, and spectral norms.
\newblock {\em SIAM Journal on Computing}, 21(1):33--42, February 1992.

\bibitem[For02]{For02}
J{\"u}rgen Forster.
\newblock A linear lower bound on the unbounded error probabilistic
  communication complexity.
\newblock {\em Journal of Computer and System Sciences}, 65(4):612--625, 2002.

\bibitem[HQ17]{HQ17}
Hamed Hatami and Yingjie Qian.
\newblock The unbounded-error communication complexity of symmetric xor
  functions.
\newblock https://arxiv.org/abs/1704.00777, 2017.

\bibitem[Pat92]{Pat92}
Ramamohan Paturi.
\newblock {On the degree of polynomials that approximate symmetric Boolean
  functions (preliminary version)}.
\newblock In {\em Proceedings of ACM Symposium on Theory of Computing}, pages
  468--474, 1992.

\bibitem[PS86]{PS86}
Ramamohan Paturi and Janos Simon.
\newblock Probabilistic communication complexity.
\newblock {\em Journal of Computer and System Sciences}, 33(1):106--123, 1986.

\bibitem[Raz03]{Raz03}
Alexander Razborov.
\newblock {Quantum communication complexity of symmetric predicates}.
\newblock {\em Izvestiya: Mathematics}, 67(1):145--159, 2003.

\bibitem[She12]{She12}
Alexander~A. Sherstov.
\newblock The multiparty communication complexity of set disjointness.
\newblock In {\em S{TOC}'12---{P}roceedings of the 2012 {ACM} {S}ymposium on
  {T}heory of {C}omputing}, pages 525--544. ACM, New York, 2012.

\bibitem[ZS09]{SZ09}
Zhiqiang Zhang and Yaoyun Shi.
\newblock Communication complexities of symmetric xor functions.
\newblock {\em Quantum Information \& Computation}, 9(3):255--263, 2009.

\end{thebibliography}

\end{document}